\documentclass[sigconf]{acmart}

\AtBeginDocument{%
  \providecommand\BibTeX{{%
    \normalfont B\kern-0.5em{\scshape i\kern-0.25em b}\kern-0.8em\TeX}}}

\setcopyright{acmcopyright}
\copyrightyear{2019}
\acmYear{2019}

\begin{document}

\title{RTM: Blockchain That Support Revocable Transaction Model}

\author{Victor Gates}
\email{victor@gateio.me}
\affiliation{%
  \institution{Gate.io Technology Inc.}
  \streetaddress{ P.O. Box 31119 Grand Pavilion, Hibiscus Way, 802 West Bay Road, Grand Cayman, KY1-1205}
  \city{Cayman Islands}
}








\renewcommand{\shortauthors}{Victor Gates}

\begin{abstract}
In many typical application scenarios, it is necessary to revoke the incorrect account operations caused by user mis-operation, financial fraud, illegal hacking, etc. Unfortunately, users often blur the lines between the concept of "transaction state revocable" and "business status revocable", which result in revocable transaction not universally supported in blockchain systems at present.

In this work, we propose \textsc{Gate}\textsc{Chain}, a blockchain that support revocable transaction model (RTM) on distributed ledger. Specifically, based on the state-of-the-art block-chain technologies, \textsc{Gate}- \textsc{Chain} can safely withdraw the account status change opera-tions by leveraging an improved account model and extra designed transaction types. On that basis, \textsc{Gate}\textsc{Chain} exploit the characteristics of functional completeness, easy to deployment and lower complexity.

\end{abstract}




\maketitle

%
%

\section{Introduction}

Blockchain technology is innovative and continuously evolving. It has gained unprecedented attention from both the research and industrial communities. As an important application of a blockchain, a distributed ledger can effectively record the change in user accounts. However, in many typical application scenarios, it is necessary to revoke the incorrect account operations caused by user mis-operation, financial fraud, hacking, etc. In particular: 

\begin{itemize}
\item In a distributed ledger, the private key of a user's account is the unique identifier of the user's identity. If the private key is stolen due to hacking or misuse, the assets held in the account cannot be recovered. A major event where this happened was on February 7, 2014, Mt.Gox, an exchange that traded 70\% of all Bitcoin at that time, lost $850,000$ bitcoins due to a hack. This event caused the exchange to go bankrupt, and to date,  has not compensated the loss to its users. 

\item In both Bitcoin\cite{ref11} and Ethereum\cite{ref12} systems, if assets are transferred to a wrong address due to any reason, technically the loss cannot be recovered. This can be caused by user's inappropriate operations, as well as malicious financial fraud. 
\end{itemize}

To ensure the safety of standard user operations, a blockchain that supports revocable operations has received widespread attention \cite{ref01}. Noteworthy, currently in both research and industry communities, some hold a traditional view believing that once transaction data is recorded on a blockchain system, it should remain immutable in any case. The "classic" Ethereum project ETC \cite{ref02} is an outstanding example. Contrary to it, those holding a development view typically argue that technology itself is continuously evolving, and the future development of technology should not be defined by the current characteristics of a blockchain \cite{ref03,ref04}. The root cause of this controversy lies in the confusion between "transaction state revocable" and  "business status revocable" on a blockchain. To be specific, from a business perspective, the irreversible nature of a blockchain transaction safely prevents data on chain from being tampered with in a decentralized manner. 

This provides a low-cost mechanism for data safety, which is considered to have broad application prospects \cite{ref05}. However, from the end-user's point of view, data modification and operations is undoubtedly a necessary function for an effective business system. Clearly, under strict regulatory conditions, a ledger system should provide the functionality of modifying and withdrawing the operation performed on the account states. 

Generally, the existing revocable technology on Bitcoin and Eth- ereum systems can be divided into three categories. The first type is represented by the EOS system \cite{ref06} which approves the revocation of operations by majority voting in quorum. However, this type of technology does not support the rollback of the transaction itself. This means that the data recorded on the blockchain cannot be changed. The second type is applying a "hard fork", which completely modifies the chain and account states. A "hard fork" can be seen as controversial by the industry. An example of this happened in June 2016, when a project named "The DAO" was hacked which resulted in a large amount of Ethereum (ETH) being stolen. This forced the Ethereum community to perform a hard fork to recover the stolen ETH \cite{ref07}. The third type uses advanced technology like the Lightning Network \cite{ref08}. The Lightning Network can cancel unconfirmed transactions, but it is not suited for a transaction that is already recorded on the blockchain. Similarly, revocable technologies on blockchains has also received widespread attention in the academic field. Specifically, some research \cite{ref09,ref10} have proposed an editable blockchain technique to modify the block by direct compliance to achieve the purpose of the revised operations. However, these technologies are mostly implemented as prototype systems. The safety and reliability of such academic systems have not yet been verified.

To address this issue, in this paper we take the first step toward realizing revocable operations on a blockchain. Targeting real-world application requirements, such as the loss of a user's key and transfer operation errors, we propose \textsc{Gate}\textsc{Chain}, a blockchain system that supports a revocable transaction model (RTM) on a distributed ledger. Specifically, based on state-of-the-art blockchain technologies, \textsc{Gate}\textsc{Chain} can safely withdraw the account status change operations by leveraging an improved account model and extra designed transaction types. On that basis, \textsc{Gate}\textsc{Chain} exploits the characteristics of functional completeness, is easy to deploy, and lowers complexity. 

To summarize, the contributions we made in this paper are as follows:

\begin{itemize}
  \item To the best of our knowledge, \textsc{Gate}\textsc{Chain} is the first work on realizing revocable account operations by leveraging an improved account model and extra designed transaction types. Considering that this blockchain provides support for transaction verification, the revocable operations proposed in this paper adopted the transaction model of a \textsc{Utxo}-based block- chain. In essence, the transaction method can make full use of proven, mature blockchain technologies, and can be regarded as a "native" realization of the revocable operations.
  \item \textsc{Gate}\textsc{Chain} proposes a user-friendly account model. Compared with the traditional blockchain account model adopted by Bitcoin and Ethereum, the account model of \textsc{Gate}\textsc{Chain} introduces a radical new vault account type, which is designed to support the withdrawal of account status change. The design of this vault account provides users with better security.
  \item A real world application of \textsc{Gate}\textsc{Chain} is provided as a verifiable implementation of revocable business logic, with more comprehensive considerations on security, usage, as well as performance issues.
\end{itemize}

The rest of the paper is organized as follows: 

Section \ref{sec2} introduces the design of RTM.

Section \ref{sec3} explains the safety considerations of RTM

Section \ref{sec4} presents the design of \textsc{Gate}\textsc{Chain}, a verifiable realization of RTM. 

Section \ref{sec5} summarizes the paper and mentions the direction of future work.

%
%

\section{The Design of RTM}\label{sec2}

Blockchain technology provides a decentralized distributed ledger architecture which supports complete, sequential transaction chain recording. In a typical distributed ledger application, the user's assets are represented as the states in the account, and the states can only be changed by designated transactions, which are recorded and managed by a blockchain system. In essence, from the application level, revoking a transaction is not deleting the records of the transaction from the blockchain, but to support the user to withdraw the account status change caused by the transaction.

In this section, we propose a Revocable Transaction Model (RTM), a blockchain-based transaction model that supports revocable logic, for the sake of avoiding user account losses caused by misuse, financial fraud, and illegal hacking. The key features of RTM can be summarized as:

\begin{itemize}

\item The intrinsic value of RTM is based on state-of-the-art block- chain technologies, which ensures that all transactions recorded on the chain cannot be tampered with.

\item To support the revocation of account status, we propose a radical new transaction model called RTM. With additionally designed transactions, RTM can support the revocation of the account status change within a designated delay period $\theta$. The design of these transactions ensure the completeness of the extended transaction state machine.

\item To support revocable transactions, RTM  includes the structure of a "\textit{Vault Account}" which is especially created for the revocation of account status. In real world applications, the vault account mechanism can tackle the problem of user private key loss, thus greatly enhance the security of the whole ledger application.

\end{itemize}

\subsection{Basic Concepts}\label{sec2-1}

Intrinsically, the basic idea of the RTM is to initiate a transaction that support withdraw the change of given account status within a certain period of time, i.e. the delay period $\theta$. To be specific, the account can be seen as the collection of business status of a given user at a given time. Consider the reliability guarantee provided by the blockchain, the transactions between accounts can be assumed as safe and reliable, with no additional design required. Therefore, the status of an account is often designated as a numerical value. The change in account assets is reflected by the change of the account status, which is represented by a real value.

Henceforth, we can give the following formal definitions of the concepts used throughout this paper.

\begin{definition}[State] \label{def01}

 As stated above, The state of a user is generally represented by a rational number $s \in R$. The state at time $n$ (i.e., the block height) is represented as $s^n$.
\end{definition}

\begin{definition}[Account] \label{def02}

 An account is defined as the collection of states. Given a state space $S$, an account space $A$, an account $A(k)\in A$ that has $k$ states can be formally defined as:

\begin{center}
$A(k)={s_1,\ldots,s_i},s_i \in S,i=[1,k]$
\end{center}

The state of $A(k)$ at time n can be represented as:

\begin{center}
$S_k^n={s_1^n,\ldots,s_i^n},s_i \in S,i=[1,k]$
\end{center}

\end{definition}

\begin{definition}[Transaction] \label{def03}
 A transaction $T^i$  represents a change of states.
\end{definition}

\begin{definition}[State Machine] \label{def04}
 A state machine represents the change of account states over time. Given state transition function $F(I,S): S \times S\rightarrow S$, where $I,S \in S$ are the input and initial state of $F(S,I)$ respectively, the state machine can be formally defined as a collection of state transition functions, $M=\lbrace F(S^i,R^i)\rightarrow S^{i+1}\rbrace$. Among them, $S^i={S_x^i}$ is the states of all account at the $i$-th transaction,  $R^i={R_x^i}$ is the input of the account at the $i$-th transaction.
\end{definition}

As stated in the above definitions, a state machine represents the relations among all transactions. In a blockchain system, transactions are grouped into blocks according to certain rules, and the blocks form a directed chain structure, which generate the basic data structure of the blockchain. In retrospect, the transaction at time $i+1$ can be formally defined as:

\begin{center}
$T^i=F(S^i,R^i) \rightarrow S^{i+1}$.
\end{center}

\begin{definition}[Ledge] \label{def05}
 A ledger is a collection of transactions. The ledger at time $n$ is represented as $L(n)={T^i},i=[1,n]$. Specifically, a ledger is an ordered set of transactions.
\end{definition}

\begin{definition}[Revocable State] \label{def06}
Definition 6 (Revocable State). From business aspect, revocable state means that the change on the account states should be partially or completely withdrawn under strict regulatory conditions. For a given transaction sequence $i$ and $j$, where $i<j$, partially revocable means that changes on some states of account $A(x)$ can be withdrawn. i.e., $s^j=s^i$,$s^i \in S_x^i$,$s^j \in S_x^j$. The complete revocable means that all states of account $A(x)$ can be withdrawn. i.e., $S_j=S_i$.
\end{definition}

\begin{definition}[Revocable Ledger] \label{def07}
A revocable ledger is a distributed ledger system which uses a blockchain that supports a complete revocable state. 
\end{definition}

\subsection{Design of Account Model}\label{sec2-2}

In this section, we introduce the basic design of the account model proposed by RTM. First, we give the formal definition of an account used in a ledger system.

\begin{definition}[Account] \label{def08}

 The account $A(k)$  of user $k$ can be formally defined as:

$A(k)= \lbrace k,addr(A),type,S(k),key(k),L(t),I(k),data \rbrace$

Where:

\begin{itemize}
\item $k$ is the unique identifier of account $A(k)$;
\item $addr(k)$ is the address of $A(k)$, i.e., an encoding of  $A(k)$'s public key;
\item $type$ represent the type of $A(k)$, $ \in {n,v}$, where "$n$" state for standard account, and "$v$" state for vault account;
\item $S(k)$ is the state managed by account $A(k)$, where $S_k^i$ represent the state of account  $A(k)$ at time $I$;
\item $key(k)$ is the private/public key pair of account $A(k)$.
\end{itemize}

\end{definition}

In order to support revocable status, RTM introduces two types of accounts, namely the standard account and the vault account. The standard account refers to the account used in traditional \textsc{Utxo}-based blockchain systems. The standard account is used for conventional account transfer transactions. In this setting, the transaction changes the account status immediately, and it does not support revocable states. 

The vault account is specifically designed for RTM. It supports revocable states. From the user's point of view, after an account transfer transaction is initiated by a vault account, the change of the account states can only come into effect after a delay period $\theta$. Within the delay period $\theta$, the user can initiate a revocable transaction to withdraw the change of the vault account states.
For vault account A(k),if type=v,then the follow holds:

\begin{itemize}
\item $L(t)$ is the clear time (block height) of account $A(k)$. 
\item $I(k)$ defines the account where the revoked state is assigned to. Namely, $I(k)$ is called as the "retrieval account" of vault account $A(k)$, and generally located on a reputable third-party exchange platform;
\item Data is a list of transactions which records all the revocable transactions initiated by vault account $A(k)$.
\end{itemize}

Generally, the vault account is designed with the following characteristics:

\begin{itemize}
\item \textbf{Time delay}: The revocable transaction $T^i=F(\lbrace S_k^i,S_j^i \rbrace ,R_k^i)$ $\rightarrow \lbrace S_k^{i+1},S_j^{i+1}\rbrace $ can only be initiated by a vault account  $A(k)$. For each revocable transaction, a delay period $\theta$ is assigned. From the application point of view, the account status change can be rolled back during the period of $\theta$.
\item \textbf{Revocable}: Within the designated delay period $\theta$, if $A(k)$ initiates a revocable transaction $T^j$  to revoke the operation of $T^i$, the revocable state can be rolled back to the original state.
\end{itemize}

\subsection{Design of Revocable Logic}\label{sec2-3}

Revocation logic represents the core of the revocable transaction model (RTM). Therefore, the additional designed transactions in RTM are detailed as follows:

\begin{definition}[$T_{VAR}$, VaultAccountRevocable Transaction] \label{def09}

Transaction $T_{VAR}$  define a revocable transaction. It is initiated by a vault account $A(k)$, implementing a tentative account transfer operation from $A(k)$ to $A(j)$. Formally:

$T_{VAR}^i (k,j,\theta)=F(\lbrace S_k^i,S_j^i \rbrace ,R_k^i) \rightarrow \lbrace S_k^{i+1},S_j^{i+1}\rbrace, \ i \in N$

Where $\theta$ is the designated delay time of $T_{VAR}^i$,and $S_k^{i+1}=S_k^i-R_k^i$, $S_j^{i+1}=S_j^i+R_k^i$.
\end{definition}

\begin{definition}[$T_R$, Revoke Transaction] \label{def10}

Transaction $T_R$ is initiated by a vault account $A(k)$ to revoke the tentative account transfer operation $T_SAR^i (k,j,\theta)$ from the original vault account $A(k)$. The revocable states of $A(k)$ is transferred to its bounded safety account $I(k)$. Formally:

$T_R^m (k,T_{VAR}^i)=F(\lbrace S_k^{i+1},S_j^{i+1}\rbrace ,I(k)) \rightarrow \lbrace S_k^i,S_j^i \rbrace , \ i,m \in N$
\end{definition}

The condition for $T_R$ to be correctly executed, attribute to that the block height of transaction T$^m$ is less than $\theta$.

\begin{definition}[$T_{VAC}$, VaultAccountCreation Transaction] \label{def11}

Transaction $T_{VAC}$ is used to create a vault account $A(k)$. $T_{VAC}$ is initiated by a normal account $A(j)$ by issuing an account transfer transaction, formally:

$T_{VAC}^i=F( \lbrace S_k^i,S_j^i,I(k) \rbrace ,R_k^i) \rightarrow \lbrace S_k^{i+1},S_j^{i+1} \rbrace, \ \ i \in N$

\end{definition}

Note that $T_{VAC}$ is not a revocable transaction. Once $T_{VAC}$ is successfully executed, the state assigned by $R_k^i$ is transferred from $A(k)$ to the newly created vault account $A(j)$, and the state of $A(j)$ is changed to $S_j^{i+1}$,  as well as the state of $A(k)$ is changed to $S_k^{i+1}$, where $S_j^{i+1}-S_k^{i+1}=R_k^i$. In the delay period $\theta$, $A(k)$ can initialize a revocable transaction $T_R$ , withdraw the assets to the retrieval account $I(k)$ bounded by vault account$ A(k)$.

Based on the designed transaction $T_{VAR}$ and $T_R$, the revocation logic of RTM is defined as follows:

\begin{enumerate}

\item At time $i$, user $k$ initiates a revocable transaction $T_{VAR}$ on vault account $A(k)$. $T_{VAR}$ specifies the tentative amount of transfer assets $I$, the destination account $A(j)$ of user $j$, and the delay period $\theta$ for the transfer to take effect.

\item After the transaction is verified and confirmed on the blockchain, $T_{VAR}$ , as a transaction, is recorded on the chain, and the status changes involved in  $T_{VAR}$ are also recorded in accounts $A(k)$ and $A(j)$, but the  state $S_k^i$ and $S_j^i$ are not immediately changed.

\item If at time $i<n<\theta$, user $k$ finds problems in the transaction, the user can initiate a revoke transaction $T_R$, and revoke the transfer amount in $T_{VAR}$ to a retrieval account $I(k)$ bounded by $A(k)$.

\item If at time $i<n<\theta$,  a revoke transaction $T_R$ targeting a $T_{VAR}$ is initiated, and $T_R$ is verified and confirmed on the blockchain, then all $T_{VAR}$ transactions initiated after time $n$ are invalid transactions.

\item At time $n=\theta$ , and $T_{VAR}$ is not revoked, the account states of both users remains unchanged. i.e., the transfer operation at time $i$ is delayed. That is, $S_k^{n+1}=S_k^n-R_k^i, S_j^{n+1}=S_j^n+R_k^i$.

\item At time $n \ge \theta$, any revoke transaction $T_R$ initiated for $T_{VAR}$ will be an invalid transaction.

\end{enumerate}

\subsection{Design of RTM's Security Model}\label{sec2-4}

\newtheorem{assumption}{Assumption}

As stated by the early work of Rosenfeld\cite{ref13}, the cost of an attack grows exponentially, which is also correct as long as the blockchain system adopted consensus in mining blocks. In this work, we adopted the formal model proposed by Garay et. al. \cite{ref14} for the core of  the Bitcoin protocol, where blocks can be found simultaneously as each mining steps. Therefore, we assume the following assumptions holds for the security model of the generalized RTM.

\begin{assumption}
The mining rate $\lambda$ of the blockchain remains constant over time. e.g. in Bitcoin, $\lambda=\frac{1}{600}$ blocks per second.
\end{assumption}

\begin{assumption}
  Under reliable networks, the block creation rates are much higher than the propagation time of blocks. Surely, the block propagate in the network is very fast relative to $\frac{1}{\lambda}$.
\end{assumption}

\begin{assumption} 
  The honest miners only keep track of the longest chain they have been presented. i.e., the longest chain prevails. 
\end{assumption}

Under these assumptions, the correct chain at time (height) h is uniquely determined as $L(n)=\lbrace B^h \rbrace$, $h=[1,t]$, where $B^h$ denotes a block at time h. Clearly, $B^h$ holds a collection of accepted transactions. Therefore, we formally define the acceptance probability of a transaction.

\begin{definition}[Acceptance Probability] \label{def12}

 The acceptance probability of a transaction that is robust against an arbitrary malicious attack. Formally, based on the above assumption, given block $B^h$ at height $h$. For a correct transaction $T_i^h$  to be record in block $B^h$, the acceptance probability can be roughly estimated as

\end{definition}

$\Pr_{accepted} (T_i^h)= \lbrace \exists i: height(T_i^h) \ge h + \frac{i}{\lambda} \rbrace$

Similarly, for a malicious transaction to be recorded in block block $B^h$, the acceptance probability can be roughly estimated as

$\Pr_{malicious} (T_i^h)= \lbrace i: h < height(T_i^h) < \infty \rbrace$

\begin{definition}[$\varepsilon$-robust Transaction] \label{def13}

A transaction is $\varepsilon$-robust transaction, or robust against malicious attack of blockchain, iff for any transactions $T_i^h$ at height $h$, the following holds:

\begin{center}
$\frac{\Pr_{malicious} (T_i^h)}{\Pr_{accepted} (T_i^h)} < \varepsilon $
\end{center}

\end{definition}

\begin{definition}[$\varepsilon$-robust Model] \label{def14}

 A transaction model is $\varepsilon$-robust Model, iff for any transactions $T_i^h$  initiated by attackers, the following holds:

\begin{center}
$\frac{\sum_{i \in T} \Pr ( \varepsilon_{mailicious} T_i^h )) }{h}  < \varepsilon$
\end{center}

\end{definition}

Without loss of generality, we adopt \cite{ref15} to define the attacker's policy as a function that determines the action of the attacker at every possible state. An attacker is assumed to initiate transaction with probability $\Pr_{malicious} (T_i^h)$. Given that each transaction is independent from other transactions, we can assume that the acceptance probability follows identical and independent distribution (i.i.d.). Thus, we simplify the probability of malicious attack risk as $\Pr_{malicious}$. Intuitively, $\Pr_{malicious}$ states the probability of a standard account being hacked.

Consider the design of RTM. With the introduction of vault account detailed in Section \ref{sec2-2}, each user can have vault account in addition to the traditional standard blockchain account. By design, vault account support cascade to form a chain of revocable account states. Figure \ref{fig01} illustrates the general design of cascaded vault account and associated RTM transactions.
 
\begin{figure}[h]
  \centering
  \includegraphics[width=\linewidth]{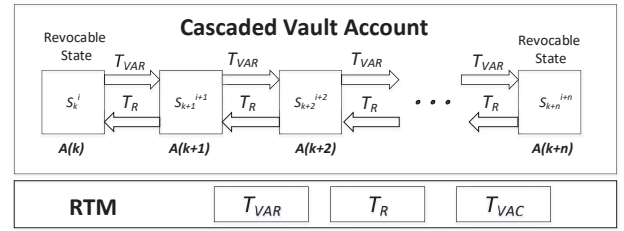}
  \caption{The design of cascade vault account in RTM security model.}
  \label{fig01}
\end{figure}

In this setting, we get the following theorem to describe the characteristics of vault account and RTM.

\begin{theorem}
 The vault account support $\varepsilon$-robust transaction.
\end{theorem}

With the strong assumption of i.i.d. among transactions, the acceptance probability of malicious attack for a vault account with n-level cascade can be minimized to 

\begin{center}
$\Pr_{malicious} (T_i^h )= \prod_{n} \Pr_{malicious} = \Pr_{malicious}^n$
\end{center}

Therefore, for $\varepsilon$-robust transactions, the probability of a vault account being malicious hacked is less than $\varepsilon^n$. Given a concrete example, if $\varepsilon=0.1$, and vault account support $5$ level cascade, then $\Pr_{malicious} (T_i^h )<0.1^5=1 \times 10^{-5}$.             

\begin{theorem}
RTM is $\varepsilon$-robust Model.
\end{theorem}

\begin{proof} Intuitively, 

\begin{center}
$\frac{\sum_h {\lbrack Pr_{malicious}\rbrack}^n}{h}  <  \varepsilon $  
\end{center}

We define that $\Pr_{malicious}=\sum_{h=0}^{\infty}  \frac{1-\Pr_{accepted}}{\Pr_{accepted}}\cdot{\lbrack \frac{\Pr_{accepted}}{1-\Pr_{accepted}} \rbrack}^n$ stochastically dominates the distribution of a random attack. The event in which the attack will accepted is then equivalent to the event that a random walk will ever arrive at a given cascaded vault account. Therefore, a necessary condition for a successful attack is that, an attacker had managed to create a chain with height $\hat{h}$, which is longer than the published chain, $\hat{h}>h$. Remind that in most cases, $\Pr_{malicious} <0.5$. This means that the acceptance probability $lim_{n \rightarrow \infty} \Pr_{malicious}^n \rightarrow 0$. 

\end{proof}

To summarize, the cascaded vault account design makes RTM a $\varepsilon$-robust Model. Theoretically, vault account can efficiently prevent malicious attack from being accepted by the blockchain.

%
%

\section{Further Security Considerations in RTM}\label{sec3}

The design of RTM is based on existing mature blockchain technologies. Noteworthy, the improved account model improves the security requirements in the actual applications, such as user key theft and transfer operation error. Additionally, RTM has introduced a radical new transaction type to tackle revocable state changes.

\subsection{The impact on transaction performance}\label{sec3-1}

To tackle the problem of user error, financial fraud and account transfer error, RTM has introduced new transaction types $T_{VAR}$ and $T_R$. Note that the implementation of $T_{VAR}$ and $T_R$ do not add additional complexity to the existing transaction model, and therefore do not degrade the overall performance of the system.

Generally, existing blockchain transaction models can be classified into two categories: Bitcoin's \textsc{Utxo}-based model and Ethereum's account model. For the \textsc{Utxo}-based model, RTM only needs to add an extra operation to check the type of the previous transaction, to ensure whether it is a revocable transaction. Remind that in order to verify the legality of the transaction, an extra operation to determine the transaction time (block height) should be added to RTM. To be specific, to query the balance of an account, it is necessary to distinguish between the revocable states and normal states in the unspent transactions. 

For Ethereum's account model, an account balance check operation should be added to the implementation of $T_{VAR}$ . To query the balance of an account, an extra account transfer operation is needed.

Overall, the extra operations needed to implement $T_{VAR}$ and $T_R$ have lower complexity , thus cannot significantly increase the cost of transaction implementation. Therefore, adding revocable transactions to the existing blockchain transaction model does not degrade the overall performance. In fact, it can only be seen as an advantage as it can prevent malicious users from attacking the system by initiating a huge number of revocable transactions.

\subsection{The impact on user operations}\label{sec3-2}

In order to tackle the problem of user key theft caused by hacking, RTM has improved the account model by introducing a radical new design named "\textit{Vault Account}". Henceforth, a user can have two types of accounts, i.e., a standard account and a vault account. Although the vault account improves the security of the account operation, it increases the complexity of the user operation. In order to avoid the security risk caused by user errors in real world implementation, the revocable account introduces extra operations for compliance verification in accordance with the business logic. Specifically, business logic compliance verification considers the following aspects:

\begin{itemize}
\item A vault account can only initialize revocable transactions $T_{VAR}$ and $T_R$, and $T_{VAR}$ / $T_R$ must also be initialized from a vault account.

\item By adopting differentiating encoding for different types of accounts, it is convenient for a user to intuitively distinguish between a retrieval account and a standard account, thus keeping users from mistakenly making a revocable transaction instead of a standard transaction.

\item A user can use the account balance to check revocable states (assets) and normal states. This helps users to confirm transaction types, and avoid blurring revocable transactions with standard transactions.
\end{itemize}

\subsection{The double spend attact}\label{sec3-3}

The "\textit{double spend}" attack means that an account can transfer the same asset to more than one account. The essence of this attack is to use some abnormal factors such as network delays to purposely cause illegal inconsistencies between the nodes in a distributed ledger system. Generally, a blockchain system can avoid the "double spend" attack by leveraging atomic broadcasting and consensus mechanisms. 

In RTM, the verification of the revocable transaction is the same as the verification of the standard transaction. To be specific, for a \textsc{Utxo}-based transaction model, the verification of a current transaction needs to specify the previous transaction. For the Ethereum-based account model, the verification needs to check the account balance. Once a transaction is issued, the assets transferred out of an account will be deducted..Therefore, revocation of the transaction will not result in additional double spend risk. Specifically, the mechanisms used by $T_{VAR}$ and $T_R$ can be detailed as :

\begin{itemize}
\item For $T_{VAR}$, the assets involved in $T_{VAR}$ cannot be transferred out multiple times. Like traditional  transfer business logic, the change of account status for a revocable transaction should be delayed. However, from the transaction point of view, as long as the transaction status change can be verified by the blockchain system, the correctness of the transaction can be guaranteed by the blockchain itself.

\item For $T_R$, the assets involved in $T_R$ cannot be withdrawn multiple times. Once a $T_R$ transaction is verified by a blockchain system, the correctness of the transaction can also be guaranteed by the blockchain itself. Furthermore, in the implementation of the business logic, a compliance rule is further defined, which defines that the receiving account cannot use the assets in a delay period θ. This compliance check further ensures the robustness of the revocable logic. 
\end{itemize}

\section{\textsc{Gate}\textsc{Chain}: A Verification implementation of RTM}\label{sec4}

Using the proposed transaction model RTM, in this paper we implemented a verification blockchain system called \textsc{Gate}\textsc{Chain}. Specifically, \textsc{Gate}\textsc{Chain} adopted the RTM design, with its structure based on the state-of-the-art platform Tendermint \cite{ref13}, leveraging the distributed network and consensus mechanism provided by Tendermint Core, and the application layer implemented by using Cosmos SDK. Figure \ref{fig02} illustrates the general structure of \textsc{Gate}\textsc{Chain}.

\begin{figure}[h]
  \centering
  \includegraphics[width=0.85\linewidth]{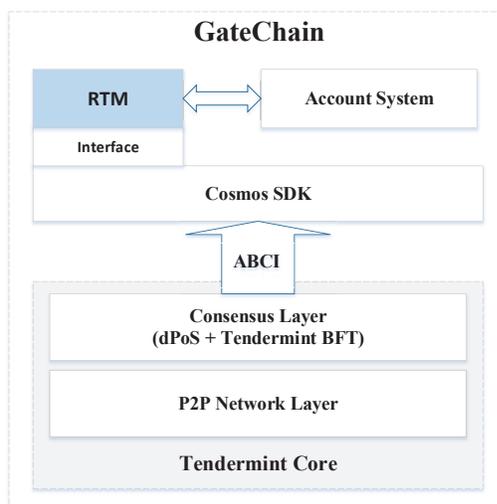}
  \caption{The general structure of \textsc{Gate}\textsc{Chain}, a verification implementation of RTM.}
  \label{fig02}
\end{figure} 

In essence, \textsc{Gate}\textsc{Chain} implemented several verification features targeting RTM. First, \textsc{Gate}\textsc{Chain} introduced the vault account to hold revocable states. Each vault account should bind a safety account as the revoked output of the vault account. Second, \textsc{Gate}\textsc{Chain} introduced the revocable transaction logic proposed by RTM. This application logic is well suited for the circumstance when a user is securely holding his private key but issues an unintended transfer transaction, for example hacking, misuse or financial fraud. Third, to support the revocable transaction logic, \textsc{Gate}\textsc{Chain} presented a solid implementation of $T_{VAR}$,$T_R$  and $T_{VAC}$ , as proposed in Section \ref{sec2-3}.

\subsection{Implementation of Account Model} \label{sec4-1}

In order to simultaneously support the revocable transactions and the standard transactions, \textsc{Gate}\textsc{Chain} has implemented two types of accounts: the standard account and the vault account.
 
The standard account has no difference with the account implemented in \textsc{Utxo}-based blockchain systems, and it is applicable to the standard account transfer scenarios. By default, the transferred assets arrive immediately, showing no revocable behavior. This is well suited for general transfer applications. 

The vault account supports revocable transactions. The assets transferred from this type of account can only arrive after the designated delay period $\theta$, while the transaction itself is adequately recorded by the blockchain system. In a single run, the status of a vault account will remain unchanged until the delay period expires. Noting that during the delay period, the vault account can initiate another transaction to revoke the aforementioned transaction, and withdraw the corresponding assets back to a designated retrieval account. 

As stated in section \ref{sec2-2}, the target account which accepts the withdrawn assets from the vault account is called a retrieval account, which can be a standard account or a vault account. It is worth noting that the withdrawn assets of the vault account are not directly returned to the originating vault account itself, but to a retrieval account its bounded to.

The introduce of retrieval account increases the security of revocable transactions. From an application point of view, distinguishing the account type makes clear to the user:

\begin{itemize}
  \item Whether the current account is a vault account or a standard account
  \item Whether the assets transferred by the account are revocable
  \item Whether the revocable transaction has taken effect. 
\end{itemize}

Therefore, this implementation lets the user select different account types according to their actual needs.

Note that a risk must be considered in real application scenario, which can be called the revoke of a delay transfer. To be specific, a malicious payer can initiate a revoke transaction which tentatively transfer certain amount to a payee. If payee mistakenly consider this as a normal transfer, and concluded the transaction in turn. The malicious payer can revoke the tentative transfer, thus causes an intended business fraud. To tackle this risk, \textsc{Gate}\textsc{Chain} introduces two encoding mechanism to technically explicify the type of transaction. Firstly, each account is encoded with a prefix to distinguish between standard account and vault account. From the prefix adhere to an account address, use can explicitly judge the type of the initiated transaction. Secondly, each transaction id is also encoded with a predefined prefix to stamp different transaction type, which details in Section 4.2.  Encoding mechanism technically prevent the user from mistaken a revoke transaction with a standard transaction, thus  take adequate action henceforth.       

\subsection{Implementation of Transaction Model}

In addition to the standard transaction types supported by a vanilla \textsc{Utxo}-based blockchain , \textsc{Gate}\textsc{Chain} introduces two types of revocable transactions named $T_{VAR}$ and $T_R$.  Both types of transactions can only be initiated by a vault account, supporting revocable states. The consensus of \textsc{Gate}\textsc{Chain} is supported by the underlying consensus mechanism provided by the Tendermint platform. The implementation of these two transactions does not add any additional cost to the system as a whole. To be specific, different types of transactions are distinguished by a set prefix encoding. That is, a transaction is named following the pattern of "fixed transaction type + transaction hash value".The transaction prefix helps the user to clearly distinguish the transaction type, avoiding *unintended* operations and invalid operations. On the other hand, this design can help the blockchain itself and the application to check the correctness of a transaction, thus improving the security of the system. Details of transaction prefixes are listed as Table \ref{tab:commands}.

\begin{table*}
  \caption{List of predefined transaction prefixes}
  \label{tab:commands}
  \begin{tabular}{ccl}
    \toprule
    Transaction type & Transaction prefix\\
    \midrule
    Irrevocable pay & INREVOCABLEPAY-<id> \\
    Revocable pay & REVOCABLEPAY-<id>\\
    Setting an account & ACCOUNTSET-<id>\\
    Revoke a revocable pay &  REVOKE-<id>\\
  Clear a Vault Account & VAULTCLEAR-<id>\\
  Regular transaction & BASIC-<id>\\
    \bottomrule
  \end{tabular}
\end{table*}

The revocable logic of the transaction is applicable to the scenario in which the user properly owns his private key, and wants to revoke the transfer operations initiated by a vault account. The logic is implemented as vault account A initiates a transaction to transfer some assets to a designated account B, and sets the transaction type as $T_{VAR}$, i.e., a revocable transaction. After the transaction is successfully executed, the transaction and the delay period are recorded in account $A$ and $B$ separately. However, the transferred assets are not immediately available in both $A$ and $B$. In the event that the transaction delay period has expired, the transfer transaction is officially effective and the assets are officially available in $B$. Therefore, during the delay period of the transaction, the vault account $A$ itself may initiate a revoke transaction $T_R$, the asset is withdrawn to a safety account bounded by $A$. Note that the operation of $T_R$  takes effect immediately, and the transfer operation is cancelled for account $B$. Generally, the safety accounts are typically located on a third-party trustworthy platform. The transaction itself is recorded in the blocks of \textsc{Gate}\textsc{Chain}, and the inactive transactions maintained in the account are cleared by \textsc{Gate}\textsc{Chain}.

To conclude, \textsc{Gate}\textsc{Chain} implements a lightweight RTM based on existing mature blockchain technologies, namely, the Tendermint platform and the Cosmos SDK.

\section{Conclusion and Outlook}\label{sec5}

Real world business, such as financial transfers, government regulation, copyright transactions, personal privacy, and transaction rollback logic, all natively require that the blockchain system should support revocable transactions. Regarding the perspective of implementation, revocable and non-tamperable are not two contradictory concepts. The reason lies in that as long as a revocable transaction is verified on the blockchain, the transaction itself is supported by the mechanism provided by a blockchain. Thus, we can add revocable features into existing blockchain systems by improving the design of the transaction model. This has proven to be the best way in natively implementing a revocable-supported blockchain.

We believe that, as an innovative technology, the blockchain ecosystem itself is continuously evolving. Some new features will be reflected in the design and implementation of the blockchain in the near future, and in the form of "native" blockchain characteristics. Undoubtedly, revocable transactions are such a feature. Although blockchain technology is still in its exploratory phase in many areas, it is certain that it will continue to evolve and improve, leading to wide adoption in many fields.

\bibliographystyle{ACM-Reference-Format}
\bibliography{gatechain-rtm}


\begin{thebibliography}{15}


\ifx \showCODEN    \undefined \def \showCODEN     #1{\unskip}     \fi
\ifx \showDOI      \undefined \def \showDOI       #1{#1}\fi
\ifx \showISBNx    \undefined \def \showISBNx     #1{\unskip}     \fi
\ifx \showISBNxiii \undefined \def \showISBNxiii  #1{\unskip}     \fi
\ifx \showISSN     \undefined \def \showISSN      #1{\unskip}     \fi
\ifx \showLCCN     \undefined \def \showLCCN      #1{\unskip}     \fi
\ifx \shownote     \undefined \def \shownote      #1{#1}          \fi
\ifx \showarticletitle \undefined \def \showarticletitle #1{#1}   \fi
\ifx \showURL      \undefined \def \showURL       {\relax}        \fi
\providecommand\bibfield[2]{#2}
\providecommand\bibinfo[2]{#2}
\providecommand\natexlab[1]{#1}
\providecommand\showeprint[2][]{arXiv:#2}

\bibitem[\protect\citeauthoryear{Accenture}{Accenture}{2016}]%
        {ref03}
\bibfield{author}{\bibinfo{person}{Accenture}.}
  \bibinfo{year}{2016}\natexlab{}.
\newblock \bibinfo{title}{Why distributed ledger technology must adapt to an
  imperfect world}.
\newblock
  \bibinfo{howpublished}{\url{https://www.accenture.com/_acnmedia/pdf-33/accenture-editing-uneditable-blockchain.pdf}}.
\newblock


\bibitem[\protect\citeauthoryear{Ateniese, Magri, Venturi, and
  Andrade}{Ateniese et~al\mbox{.}}{2017}]%
        {ref09}
\bibfield{author}{\bibinfo{person}{Giuseppe Ateniese},
  \bibinfo{person}{Bernardo Magri}, \bibinfo{person}{Daniele Venturi}, {and}
  \bibinfo{person}{Ewerton Andrade}.} \bibinfo{year}{2017}\natexlab{}.
\newblock \showarticletitle{Redactable blockchain--or--rewriting history in
  bitcoin and friends}. In \bibinfo{booktitle}{\emph{2017 IEEE European
  Symposium on Security and Privacy (EuroS\&P)}}. IEEE,
  \bibinfo{pages}{111--126}.
\newblock


\bibitem[\protect\citeauthoryear{Birch}{Birch}{2016}]%
        {ref01}
\bibfield{author}{\bibinfo{person}{Dave Birch}.}
  \bibinfo{year}{2016}\natexlab{}.
\newblock \bibinfo{title}{Mutable and immutable blockchains}.
\newblock
  \bibinfo{howpublished}{\url{https://www.chyp.com/mutableand-immutable-blockchains/}}.
\newblock


\bibitem[\protect\citeauthoryear{et~al.}{et~al.}{2018}]%
        {ref06}
\bibfield{author}{\bibinfo{person}{Greg~Lee. et al.}}
  \bibinfo{year}{2018}\natexlab{}.
\newblock \bibinfo{title}{EOS.IO Technical White Paper v2}.
\newblock
  \bibinfo{howpublished}{\url{https://github.com/EOSIO/Documentation/blob/master/TechnicalWhitePaper.md/}}.
\newblock


\bibitem[\protect\citeauthoryear{Ethereum}{Ethereum}{2018}]%
        {ref02}
\bibfield{author}{\bibinfo{person}{Ethereum}.} \bibinfo{year}{2018}\natexlab{}.
\newblock \bibinfo{title}{Ethereum Classic Project}.
\newblock \bibinfo{howpublished}{\url{https://ethereumclassic.github.io/}}.
\newblock


\bibitem[\protect\citeauthoryear{Garay, Kiayias, and Leonardos}{Garay
  et~al\mbox{.}}{2015}]%
        {ref14}
\bibfield{author}{\bibinfo{person}{Juan Garay}, \bibinfo{person}{Aggelos
  Kiayias}, {and} \bibinfo{person}{Nikos Leonardos}.}
  \bibinfo{year}{2015}\natexlab{}.
\newblock \showarticletitle{The bitcoin backbone protocol: Analysis and
  applications}. In \bibinfo{booktitle}{\emph{Annual International Conference
  on the Theory and Applications of Cryptographic Techniques}}. Springer,
  \bibinfo{pages}{281--310}.
\newblock


\bibitem[\protect\citeauthoryear{Greenspan}{Greenspan}{2017}]%
        {ref04}
\bibfield{author}{\bibinfo{person}{Gideon Greenspan}.}
  \bibinfo{year}{2017}\natexlab{}.
\newblock \bibinfo{title}{The Blockchain Immutability Myth}.
\newblock
  \bibinfo{howpublished}{\url{https://www.coindesk.com/blockchain-immutability-myth/}}.
\newblock


\bibitem[\protect\citeauthoryear{Mehar, Shier, Giambattista, Gong, Fletcher,
  Sanayhie, Kim, and Laskowski}{Mehar et~al\mbox{.}}{2019}]%
        {ref07}
\bibfield{author}{\bibinfo{person}{Muhammad~Izhar Mehar},
  \bibinfo{person}{Charles~Louis Shier}, \bibinfo{person}{Alana Giambattista},
  \bibinfo{person}{Elgar Gong}, \bibinfo{person}{Gabrielle Fletcher},
  \bibinfo{person}{Ryan Sanayhie}, \bibinfo{person}{Henry~M Kim}, {and}
  \bibinfo{person}{Marek Laskowski}.} \bibinfo{year}{2019}\natexlab{}.
\newblock \showarticletitle{Understanding a revolutionary and flawed grand
  experiment in blockchain: the DAO attack}.
\newblock \bibinfo{journal}{\emph{Journal of Cases on Information Technology
  (JCIT)}} \bibinfo{volume}{21}, \bibinfo{number}{1} (\bibinfo{year}{2019}),
  \bibinfo{pages}{19--32}.
\newblock


\bibitem[\protect\citeauthoryear{Nakamoto}{Nakamoto}{2019}]%
        {ref11}
\bibfield{author}{\bibinfo{person}{Satoshi Nakamoto}.}
  \bibinfo{year}{2019}\natexlab{}.
\newblock \bibinfo{booktitle}{\emph{Bitcoin: A peer-to-peer electronic cash
  system}}.
\newblock \bibinfo{type}{{T}echnical {R}eport}. \bibinfo{institution}{Manubot}.
\newblock


\bibitem[\protect\citeauthoryear{Poon and Dryja}{Poon and Dryja}{2016}]%
        {ref08}
\bibfield{author}{\bibinfo{person}{Joseph Poon} {and} \bibinfo{person}{Thaddeus
  Dryja}.} \bibinfo{year}{2016}\natexlab{}.
\newblock \bibinfo{title}{The bitcoin lightning network: Scalable off-chain
  instant payments}.
\newblock
\newblock


\bibitem[\protect\citeauthoryear{Puddu, Dmitrienko, and Capkun}{Puddu
  et~al\mbox{.}}{2017}]%
        {ref10}
\bibfield{author}{\bibinfo{person}{Ivan Puddu}, \bibinfo{person}{Alexandra
  Dmitrienko}, {and} \bibinfo{person}{Srdjan Capkun}.}
  \bibinfo{year}{2017}\natexlab{}.
\newblock \showarticletitle{$\mu$chain: How to Forget without Hard Forks.}
\newblock \bibinfo{journal}{\emph{IACR Cryptology ePrint Archive}}
  \bibinfo{volume}{2017} (\bibinfo{year}{2017}), \bibinfo{pages}{106}.
\newblock


\bibitem[\protect\citeauthoryear{Rosenfeld}{Rosenfeld}{2014}]%
        {ref13}
\bibfield{author}{\bibinfo{person}{Meni Rosenfeld}.}
  \bibinfo{year}{2014}\natexlab{}.
\newblock \showarticletitle{Analysis of hashrate-based double spending}.
\newblock \bibinfo{journal}{\emph{arXiv preprint arXiv:1402.2009}}
  (\bibinfo{year}{2014}).
\newblock


\bibitem[\protect\citeauthoryear{Sapirshtein, Sompolinsky, and
  Zohar}{Sapirshtein et~al\mbox{.}}{2016}]%
        {ref15}
\bibfield{author}{\bibinfo{person}{Ayelet Sapirshtein},
  \bibinfo{person}{Yonatan Sompolinsky}, {and} \bibinfo{person}{Aviv Zohar}.}
  \bibinfo{year}{2016}\natexlab{}.
\newblock \showarticletitle{Optimal selfish mining strategies in bitcoin}. In
  \bibinfo{booktitle}{\emph{International Conference on Financial Cryptography
  and Data Security}}. Springer, \bibinfo{pages}{515--532}.
\newblock


\bibitem[\protect\citeauthoryear{Wikipedia}{Wikipedia}{2018}]%
        {ref05}
\bibfield{author}{\bibinfo{person}{Wikipedia}.}
  \bibinfo{year}{2018}\natexlab{}.
\newblock \bibinfo{title}{Blockchain}.
\newblock
  \bibinfo{howpublished}{\url{https://en.wikipedia.org/wiki/Blockchain}}.
\newblock


\bibitem[\protect\citeauthoryear{Wood et~al\mbox{.}}{Wood
  et~al\mbox{.}}{2014}]%
        {ref12}
\bibfield{author}{\bibinfo{person}{Gavin Wood} {et~al\mbox{.}}}
  \bibinfo{year}{2014}\natexlab{}.
\newblock \showarticletitle{Ethereum: A secure decentralised generalised
  transaction ledger}.
\newblock \bibinfo{journal}{\emph{Ethereum project yellow paper}}
  \bibinfo{volume}{151}, \bibinfo{number}{2014} (\bibinfo{year}{2014}),
  \bibinfo{pages}{1--32}.
\newblock


\end{thebibliography}

\end{document}